\documentclass{llncs}

\usepackage{microtype}
\usepackage{cite}

\usepackage{amsmath}
\usepackage{amssymb}

\usepackage{graphicx}

\usepackage{subfig}
\usepackage{wrapfig}

\spnewtheorem{observation}[theorem]{Observation}{\bfseries}{\itshape}
\spnewtheorem*{uremark}{Remark}{\bfseries}{\rm}

\usepackage[usenames,dvipsnames]{color}

\newcommand{\Vbest}{V_{\mathrm{best}}}
\newcommand{\Vapp}{V_{\mathrm{apx}}}
\newcommand{\wapp}{w_{\mathrm{apx}}}
\newcommand{\wopt}{w_{\mathrm{opt}}}
\let\eps\varepsilon
\newcommand{\A}{\mathcal{A}}

\DeclareMathOperator{\conv}{CH}
\DeclareMathOperator{\dist}{dist}
\DeclareMathOperator{\width}{width}

\title{How to Cover a Point Set with a V-Shape of Minimum Width%
  \thanks{Work on this paper has been supported by grant No.~2006/194
    from the U.S.-Israel Binational Science Foundation and by NSF
    Grant CCF-08-30691.  Work by Boris Aronov has also been supported
    by NSA MSP Grant H98230-10-1-0210.  An extended abstract of this
    paper appeared in the \emph{Proceedings of the Algorithms and Data
      Stuctures Symposium (WADS'11)} \cite{AD11}.}}

\author{Boris Aronov \and Muriel Dulieu}

\institute{Department of Computer Science and Engineering, Polytechnic 
 Institute of NYU, Brooklyn, NY~11201-3840, USA;
 \email{aronov@poly.edu}, \email{mdulieu@gmail.com}}

\begin{document}
\maketitle
\pagestyle{plain}
\thispagestyle{plain}
\begin{abstract}
  A balanced V-shape is a polygonal region in the plane contained in
  the union of two crossing equal-width strips.  
  It is delimited by two pairs of parallel rays that emanate from two points
  $x$, $y$, are contained in the strip boundaries, and are mirror-symmetric with respect to the line $xy$.
   The width of a balanced V-shape is the width of the strips.
   
  We first present an $O(n^2 \log n)$ time algorithm to compute, given
  a set of $n$~points $P$, a minimum-width balanced V-shape covering
  $P$.
  We then describe a PTAS for computing a $(1+\eps)$-approximation of
  this V-shape in time $O((n/\eps)\log n +
  (n/\eps^{3/2})\log^2(1/\eps))$.
  A much simpler constant-factor approximation algorithm is also described.
\end{abstract}

\section{Introduction}
\label{sec:introduction}

\paragraph*{Motivation.}
The problem we consider in this paper was motivated by the following
curve reconstruction question: One is given a set of points sampled
from a curve in the plane.  The sample is noisy in the sense that
the points lie near the curve, but not exactly on it.  One would like to
reconstruct the original curve from this data.  Clearly one has to
make some assumptions about the point set and the curve: If the curve is
``too wiggly'' or the noise is too large, little can be done.  One
approach is to assume that the curve is smooth and the sample points
lie not too far from it; see \cite{curve-noisy,curve-dey} and
references therein.\footnote{%
  See \cite{alt-guibas-survey} for a detailed survey of different
  notions of measuring similarity between geometric objects;
  is there a sensible (and relevant for our purposes) notion of
  closeness between a discrete unordered point set and a curve?}
Roughly speaking, one can then approximate a stretch of a curve by an
elongated rectangle (or strip) whose width is determined both by the
curvature of the curve and the amount of noise.  Refining this
approximation allows one to reconstruct the location of the curve and
its normal vector.

Complications arise when a curve makes a sharp turn, as it does not
have a well-defined direction near the point of turn.  It has been
suggested \cite{curve-noisy,proj-clustering} that one approach to handle this situation
is to replace fitting the set of points corresponding to a smooth arc of a
curve with a strip by fitting with a wedge-like shape that we call a ``balanced
V-shape;'' perhaps one might incorporate it in an algorithm such as
that of
\cite{Funke-Ramos}.  It is meant to model one thickened turn in a 
piecewise-linear curve; refer to the figure and precise
definitions below.

In this paper, we construct a slower exact algorithm for identifying a V-shape that best
fits a given set of points in the plane, then a
faster constant-factor approximation algorithm, and finally a
considerably more involved algorithm that produces a
$(1+\eps)$-approximation, for any positive~$\eps$.

The problem we solve is a new representative of a widely studied class
of problems, namely \emph{geometric optimization} or \emph{fitting}
questions; see
\cite{coreset-survey,random-opt-survey,alg-opt-survey,eff-alg-opt-survey}
and references therein.  Generally, the problem is to find a shape
from a given class that best fits a given set of points.  Classical
examples of such problems are linear regression in statistics, the
computation of the width of a point set (which constructs a
minimum-width strip covering the set), computing a minimum enclosing
ball, cylinder, or ellipsoid, a minimum-width spherical or
cylindrical shell, or a small number of strips of minimum width,
covering the point set; see~\cite{chan-apx-all,coreset-survey}.

Previous work most closely related to our problem is that of Glozman,
Kedem, and Shpitalnik \cite{GKS}.  They compute a double-ray center
for a planar point set $S$.  A~\emph{double-ray center} is a pair of
rays emanating from a common apex, minimizing the Hausdorff distance
between $S$ and the double ray.  While the shape they consider is not
exactly a V-shape, it is similar enough to be used for the same
purpose.  The exact algorithm they present runs in $O(n^3 \alpha (n)
\log^2 n)$ time, however, in contrast to our near-quadratic-time
algorithm; here $\alpha(n)$ is the inverse Ackermann function.

Another paper closely related to our problem is that of Agarwal,
Procopiuc, and Varadarajan \cite{2-line-center}.
It concerns the
2-line-center problem studied extensively in the past; see the
references in \cite{2-line-center}.  The goal is to cover a given set
of points by two strips of minimum possible width.  One
application is fitting \emph{two} lines to a point set.  There had
been several previously known near-quadratic-time exact algorithms for
the problem.  An $O(n\log n)$-time 6-approximation algorithm, and an
$O(n \log n + n\eps^{-2}\log(1/\eps)+ \eps^{-7/2} \log(1/\eps))$-time
$(1+\eps)$-approximation algorithm were presented in
\cite{2-line-center}.  A V-shape covering a point set is a special
case of covering a point set by two strips, so some of the
tools from \cite{2-line-center} apply to our problem as well.

\paragraph*{Problem statement and results.}
\begin{wrapfigure}[9]{r}{0.47\textwidth}
  \scalebox{0.4}{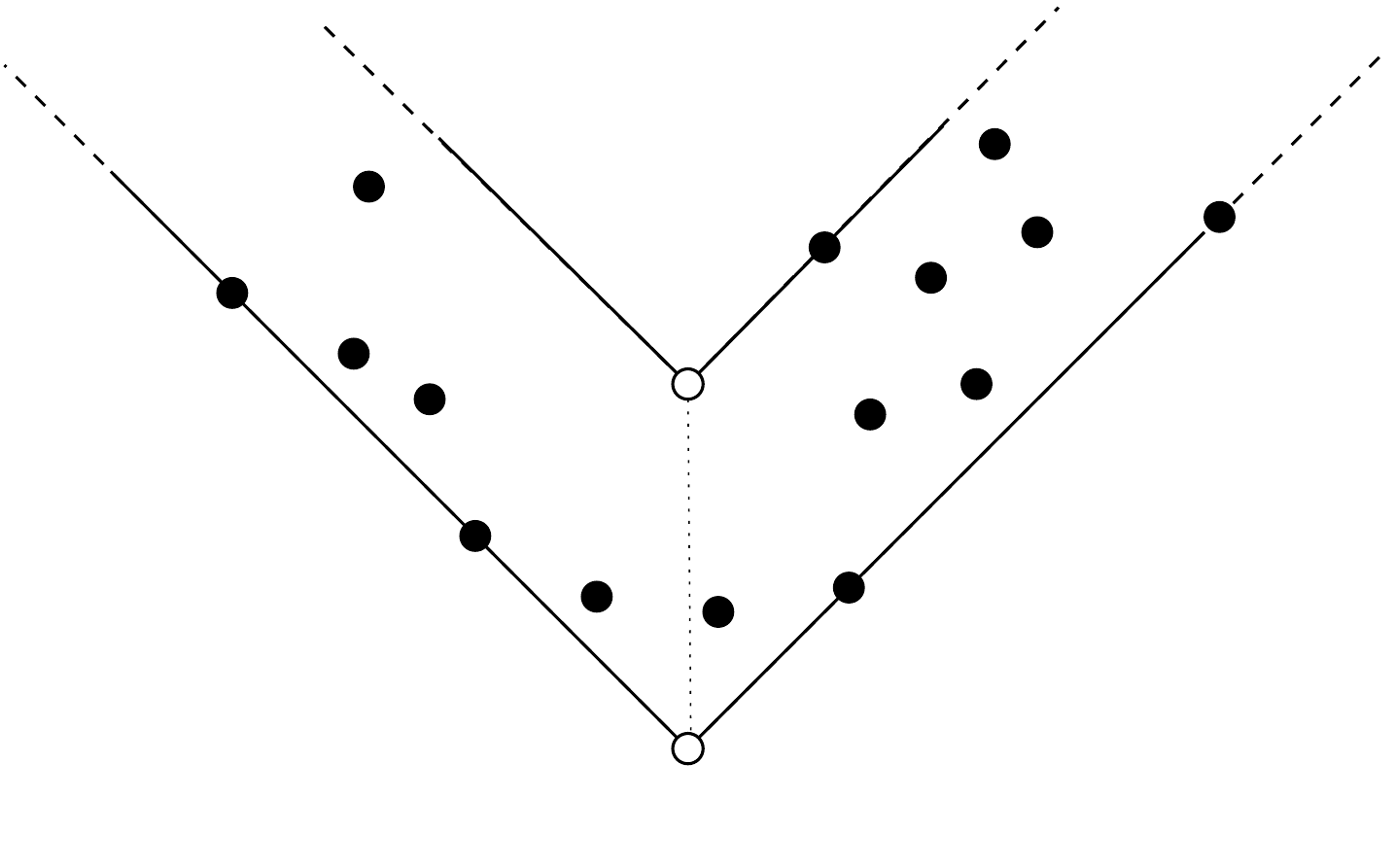}  
\end{wrapfigure}
In this paper, we focus on the class of polygonal regions in the plane
that we call balanced V-shapes.   
A \emph{balanced V-shape} has 
two \emph{vertices} $x$ and $y$ and is delimited by two pairs of
parallel rays. 
One pair of 
parallel rays emanate from $x$ and $y$ on one side of
the line $xy$ and the other pair of rays emanate from $x$ and $y$ on
the other side of $xy$, symmetrically with respect to $xy$  
(see the above figure).
%
%
%
%
In particular, a balanced V-shape is completely contained in the union
of two crossing strips of equal width.  Its \emph{width} is the width
of the strips.

Consider a point set $P$ of $n$ points in the plane.  
We describe, in Section~\ref{sec:algorithm}, an $O(n^2\log n)$ time 
algorithm that computes a balanced V-shape with minimum width covering $P$. 

Our algorithm actually identifies a particular type of V-shapes that
we call ``canonical'' (see below for definitions) and enumerates all
minimum-width canonical V-shapes covering~$P$; as some degenerate
$n$-point sets have $\Theta(n^2)$ such V-shapes (
see Section~\ref{sec:lower-bound}), this approach will probably not yield
a subquadratic algorithm.  This leaves open the problem of how quickly
one can identify just \emph{one} minimum-width V-shape covering $P$.

In Section~\ref{sec:13-approx}, we present an $O(n \log n)$ algorithm
that constructs a V-shape covering $P$ with width at most 13 times the
minimum possible width.  In Section~\ref{sec:one+eps}, we show how to
construct a $(1+\eps)$-approximation in time $O((n/\eps)\log n +
(n/\eps^{3/2})\log^2(1/\eps))$, starting with the 13-approximation
obtained earlier.

\section{Reduction to canonical V-shapes}
\label{sec:canonical}

In the remainder of this paper, for simplicity of presentation and 
unless noted otherwise, we assume
that the points of $P$ are \emph{in general position}: no three points are
collinear and no two pairs of points define parallel lines. 
All algorithms can be adapted to degenerate inputs without asymptotic slowdown.

We will find it convenient to consider a larger class of objects,
namely V-shapes. 
A (\emph{not necessarily balanced}) \emph{V-shape} (refer to the
figure below) is a polygonal
\begin{wrapfigure}[9]{r}{0.45\textwidth}
  \scalebox{0.4}{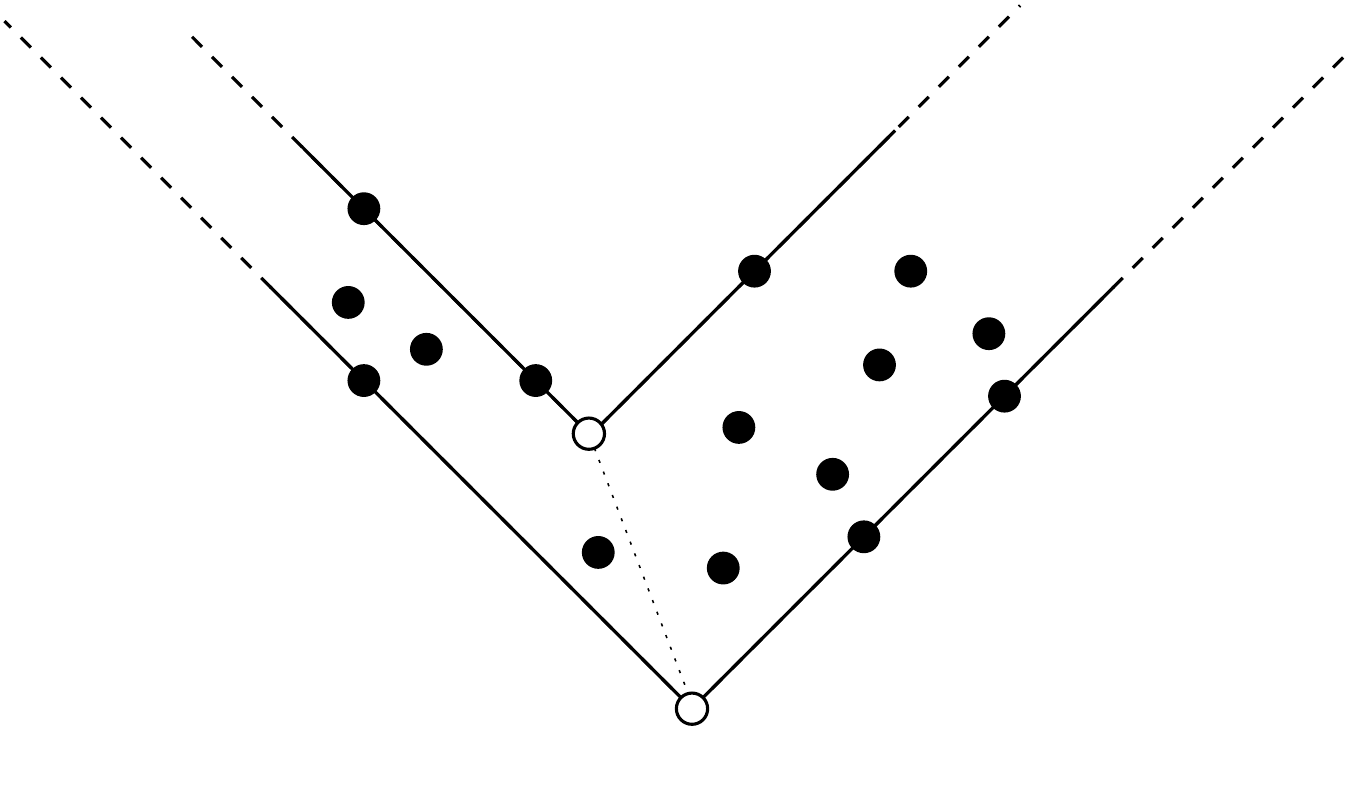}  
\end{wrapfigure}
region similar to a balanced V-shape except that the widths of its two
arms need not be the same.  More formally, a \emph{V-shape}~$V$ is a
polygonal region bounded by two pairs of parallel rays emanating from
its two \emph{vertices} $x$ and $y$.  One pair of parallel rays
(\emph{left rays} $X_1$ and $Y_1$) lies on the left side of the
directed line $xy$, while the other pair (\emph{right rays} $X_2$ and
$Y_2$) lies on its right side.  The \emph{inner rays} $X_i$ emanate
from $x$, while \emph{outer rays} $Y_i$ emanate from $y$.  $X_1\cup
X_2$ is the \emph{inner boundary} of $V$, while $Y_1\cup Y_2$ is its
\emph{outer boundary}.  The \emph{left arm} of $V$, $L=L(V)$, is its
portion on the left of $xy$; i.e., it is the region bounded by rays
$X_1$ and $Y_1$ and segment $xy$.  The \emph{width} of the left arm,
$\width(L(V))$, is the distance between $X_1$ and $Y_1$.  The right
arm and its width are defined analogously.  The \emph{width} of $V$,
$\width(V)$, is the larger of the widths of its two arms.  $V$ is
contained in the union of two strips $S_1$ and $S_2$: $S_i$ is
delimited by the lines containing $X_i$ and $Y_i$, respectively; we
refer to $S_1$ and $S_2$ as the \emph{left} and \emph{right}
\emph{strip} of $V$, respectively.

A minimum-width balanced V-shape can be obtained from a minimum-width 
V-shape by widening the narrower arm until the widths of the arms are equal.

In the remainder of the paper, the $n$-point set $P$ is fixed.  To
avoid trivial cases, we assume that $n\geq 5$.  By the general
position assumption, all points of $P$ cannot be collinear, nor can
$P$ be covered by a V-shape of zero width. 
We need not consider V-shapes with all points in one strip as according
to Lemma~\ref{lem:no-empty-strip}, 
such a V-shape does not have minimum width.
\begin{lemma}
  \label{lem:no-empty-strip}
  In a positive-width minimum-width V-shape $V$ covering $P$, it is
  not possible that one of the strips already contains $P$ in its
  entirety.
\end{lemma}
\begin{proof}
  Suppose $S_1:=S_1(V)$ covers $P$.  Let $w>0$ be its width.  We argue
  that there is a V-shape covering $P$ of width $w/2+\eps$, for any
  positive $\eps$, so $V$ does not have minimum width.  Indeed, let
  $\ell$ be the median line of $S_1$.  Cut $S_1$ by $\ell$ into two
  parallel strips of width $w/2$.  They cover $P$.  They do not form a
  V-shape, but they can be approximated arbitrarily closely by a
  V-shape near $P$, by placing its vertex~$y$ sufficiently far to the
  left of $P$ along $\ell$, $x$ at the rightmost point of $\conv(P) \cap
  \ell$, and the boundary rays near-parallel to $\ell$. \hfill $\Box$
\end{proof}

Unless otherwise stated, the only particular V-shapes we will be interested in are the ones we call canonical. A V-shape is \emph{canonical}, if the bounding rays of each arm pass through exactly
three points of $P$; more precisely if $|X_i\cap P|+|Y_i\cap P|\geq3$,
for $i=1,2$ (recall that, by our general position assumption,
$|X_i\cap P|,|Y_i\cap P|\leq 2$); in addition, we require that each arm
of a canonical V-shape covering~$P$ is locally of minimum width, i.e.,
neither arm can be narrowed by an infinitesimal motion.

The reason why we consider only canonical V-shapes is that 
at least one minimum-width V-shape covering $P$ is canonical (see
Lemma~\ref{CanonicalV-shape} below), so we can 
confine the search to canonical V-shapes and 
discard any non-canonical V-shapes considered by our algorithm.
%

\begin{lemma}
\label{CanonicalV-shape}
At least one minimum-width V-shape covering $P$ is canonical.
\end{lemma}

\begin{proof}

In order to prove this lemma, we first explain why we can assume
that $|X_i\cap P|+|Y_i\cap P|\geq3$.  We then discuss how
the boundary points may be positioned on the arms.

By Lemma~\ref{lem:no-empty-strip}, in no minimum-width covering V-shape
one strip covers~$P$ completely.
Hence in the following we disregard this
possibility.

We present a sequence of transformations, starting with a minimum-width
V-shape covering~$P$, which do not
increase its width and end in a canonical V-shape.  We begin by
translating its outer boundary in the direction of first $Y_1$ and
then $Y_2$ to ensure that each of the outer rays contains a point of $P$;
this point might be $y$.  Now translating the inner boundary, first in
the direction opposite to that of $X_1$, and then that of $X_2$, we
guarantee that each of $X_1$, $X_2$ contains a point of $P$.  (By
Lemma~\ref{lem:no-empty-strip}, an outer ray $Y_i$ cannot meet its
corresponding inner ray $X_i$ without meeting a point of $P$, as we
started with a minimum-width V-shape.)  Now consider an arm (say,
left) of the resulting V-shape.  We will further transform it so that
$|X_1 \cap P| + |Y_1 \cap P|>2$.  We have so far ensured that each of
$X_1$ and $Y_1$ contains at least one point.  If exactly one point is
present on each left ray, $X_1$ and $Y_1$ can be rotated so that the
width of $S_1$ shrinks.  
This process stops either when $S_1$ collapses to a line (in which
case it is easy to check that $Y_1$ contains two points of $P$ and
$X_1$ contains at least one) or when three points $s_1,s_2,s_3$
lie in $X_1 \cup Y_1$, say $s_1$ and $s_2$ on one ray and $s_3$ on the
other.  In the latter case, unless the angles $\angle s_3 s_1 s_2$ and $\angle s_1 s_2
s_3$ are acute, the rotation can proceed, further narrowing $S_1$.
Corollary~\ref{Cor:angle} follows from this last condition.

Repeating the process with the right arm, we arrive at a covering
V-shape whose width is no larger than that of the original V-shape,
with the property that (a)~it satisfies Corollary~\ref{Cor:angle} if there is no zero-width arm,
(b)~each bounding ray contains a point of $P$, and (c)~$|X_1
\cap P| + |Y_1 \cap P| \geq 3$ and $|X_2 \cap P| + |Y_2 \cap P| \geq
3$; i.e., the resulting covering V-shape is canonical and as good or
better in terms of width.  Hence, indeed, it is sufficient to examine
only canonical V-shapes.  \hfill $\Box$

\end{proof}

\begin{corollary}
  \label{Cor:angle}
  Let $P$ be a point set and $V$ be a minimum-width canonical V-shape
  covering it.  Let $s_1, s_2, s_3 \in X$ be the points on the
  boundary of a non-zero-width arm of $V$, with $s_1$ and $s_2$ on one
  ray and $s_3$ on the other.  Then the angles $\angle s_3 s_1 s_2$
  and $\angle s_1 s_2 s_3$ are acute.
\end{corollary}


%
%
All canonical minimum-width V-shapes fall into the following three
categories: 
\begin{description}\itemsep 0pt \parsep 0pt \parskip 0pt
\item[both-outer] Each outer ray contains exactly two points of $P$, and each
  inner ray contains at least one. 
\item[inner-outer] On one arm of the V-shape, the inner ray contains exactly two points of $P$; on the other
  arm, the outer ray contains exactly two points. The other rays contain at least one point of $P$.
\item[both-inner] Each inner ray contains exactly two points of $P$ and each outer 
  ray contains at least one.
\end{description}

\section{Computing a canonical minimum-width V-shape}
\label{sec:algorithm}

To find a canonical minimum-width V-shape covering $P$, we will search
independently for the best solution for each of the three types identified
above and output the V-shape that minimizes the width.  Let $H$ be the
convex hull of $P$.

\paragraph*{V-shapes of both-outer type.}

Consider a covering V-shape~$V$ with outer rays $Y_1,Y_2$
\begin{wrapfigure}[8]{r}{0.5 \textwidth}
  \scalebox{0.38}{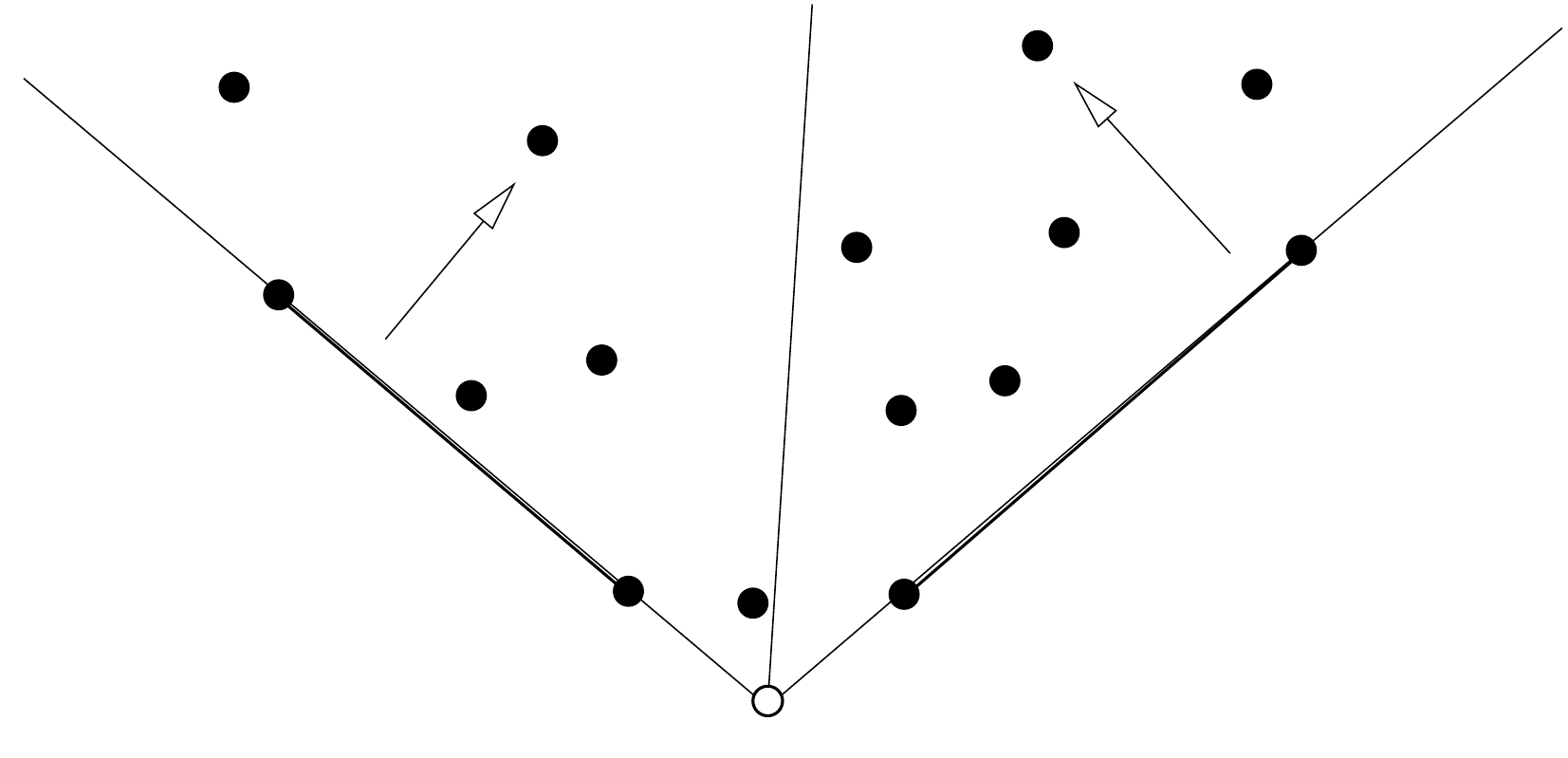}
\end{wrapfigure}
containing edges $e_1,e_2$ of $H$, respectively; refer to the figure
on the right.  Let $\ell$ be the bisector of the angle $\angle Y_1yY_2$. 
Notice that $V$ is not minimal unless its width is given by the
largest distance from a point in $P$ to its closest outer ray.
Therefore, we can assume that points of $P$ left of $\ell$ belong to the
left arm of $V$ and points right of $\ell$---to its right arm.


Thus, given $Y_1$, $Y_2$, and $\ell$ it is sufficient to determine the
furthest point from~$Y_1$ to the left of $\ell$ and the furthest point
from $Y_2$ to the right of $\ell$.  The larger distance determines the
width of $V$.  This can be accomplished by building a data structure
$D(P)$ on $P$ that supports the following queries: Given a
halfplane~$h$ and a direction~$d$, return an extreme point of $P \cap
h$ in direction $d$.  $O(n^2)$ queries are sufficient to enumerate all
choices of $e_1,e_2$ and identify the best both-outer-type V-shape.
$D(P)$ can be constructed in $O(n^2 \log n)$ time and supports
logarithmic-time queries, resulting in total running time of
$O(n^2\log n)$.

$D(P)$ is constructed as follows: We build the arrangement
$\A=\A(P^*)$ of lines dual to points of $P$.  Cells of $\A$ correspond
to different ways to partition $P$ by a line.  We construct a directed
spanning tree $T$ of the cells of $\A$, starting with the bottommost
cell and allowing only arcs from a cell~$f$ to a cell immediately
above~$f$ and sharing an edge with it; we use $P_f \subset P$ to denote
the convex hull of the set of points whose dual lines lie below $f$.  Using $T$ as the
history tree, we store the convex hull $P_f$ for every face $f \in
A$, in a fully persistent version~\cite{Bob} of the semi-dynamic
convex hull data structure of~\cite{Preparata}.  We also preprocess
$\A$ for point location.  Given a query (say, upper) half-plane $h$ and
direction~$d$, we locate the face $f$ of $\A$ containing the point
dual to the bounding line of $h$ and consult the data structure
associated with $f$ and storing $P_f=P\cap h$ to find the extreme
point of $P_f$ in direction~$d$, all in logarithmic time.  
\paragraph{V-shapes of inner-outer type.}


In this section, we describe how to find a minimum-width canonical
V-shape covering $P$ and having exactly one edge of $\conv(P)$, say~$e$,
on its outer boundary; it
contains two points of $P$ on the inner bounding ray of its other arm.
We handle each choice of $e$ independently, in $O(n \log n)$ time,
yielding overall $O(n^2\log n)$ running time.

\begin{wrapfigure}[15]{r}{0.45\textwidth}
 \includegraphics[scale=0.85]{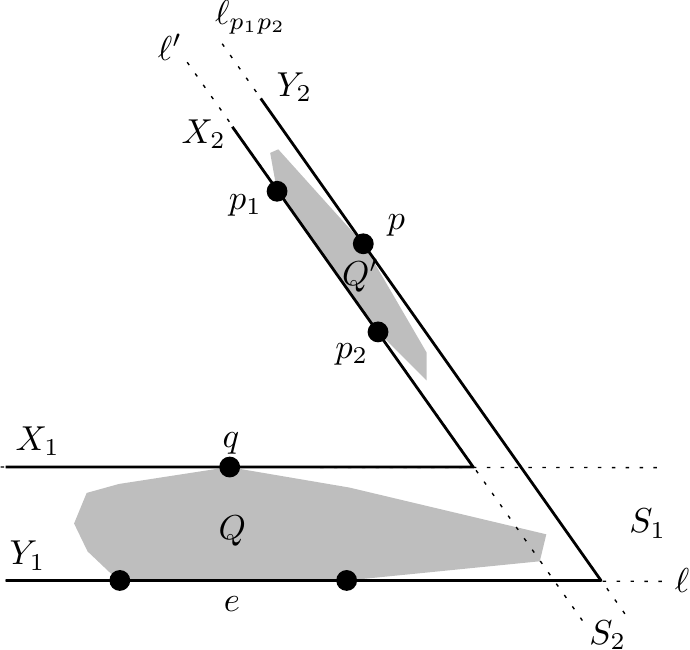}
\end{wrapfigure}

Having fixed an edge $e$ of $\conv(P)$, consider a (minimum-width
canonical) V-shape~$V$ covering~$P$ that has $e$ on its boundary. For
ease of description, suppose $Y_1 \supset e$, $X_2$ contains two
points $p_1,p_2 \in P$, while both $Y_2$ and $X_1$ contain at least one point
of $P$ each, denoted $p$ and $q$, respectively; see the figure on the right.

Let $\ell$ be the line containing $e$, and $\ell'$ be the line
containing $e':=p_1p_2$.  Set $Q:=S_1\cap P$ and $Q':=P\setminus Q$.  We
observe that
\begin{enumerate} \itemsep 0pt \parsep 0pt \parskip 0pt
\item $Q$ is the set of points of $P$ at distance at most
  $\dist(q,\ell)$ from $\ell$;
\item $p_1p_2$ is an edge of $\conv(Q')$;
\item $Y_2$ is contained in a supporting line $\ell_{p_1p_2}$ of
  $\conv(Q')$ (which must also be a supporting line to $\conv(P)$ for 
  $V$ to cover $P$) parallel to $\ell'$; this line lies on
  the same side of $\ell'$ as $Q'$;\footnote{%
    If $\width(R(V))=0$, we have $Q' \subset \ell'$ and $Y_2 \subset
    \ell'$.}
   and
\item $\width(V)=\max(\width(S_1),\width(S_2))=
  \max(\dist(q,\ell),\dist(\ell' ,\ell_{p_1p_2}))$.
\end{enumerate}
Our algorithm enumerates all choices for the point $q$, in order of
decreasing distance from $\ell$. 
  For the current choice of $q$, it
maintains (the boundary of)
$\conv(Q')$, say as an AVL~tree, and, for each edge $e'$ of
$\conv(Q')$, 
the distance from $\ell'$ to the furthest point of $CH(P)$ to the
right of (i.e., ``beyond'') $\ell'$.
Edges with distances are
stored in a min-heap; the minimum such distance gives the minimum
width for $S_2$ for the current choice of $S_1$; the larger of the two widths
determines the width of the current V-shape.  We record the best width
of any V-shape encountered in the process.

The algorithm is initialized with the set $Q'$ 
containing the two points of $P$ furthest from $\ell$ (the case where $Q'$ contains 
only one point treated by the both-outer case as the zero-width strip $S_2$ can be rotated until it 
contains one edge of $\conv(P)$).  A generic step of the
algorithm involves moving the current point $q$ from $Q$ to $Q'$.  We
update the convex hull of $Q'$ by computing the supporting tangents
from $q$ to the old hull, in $O(\log n)$ time.  For the two new hull
edges $e_1$, $e_2$, we compute the corresponding supporting lines
$\ell_{e_1},\ell_{e_2}$ of $\conv(P)$, using a suitable balanced-tree
representation of $\conv(P)$, also in logarithmic time.  We add the new
edges with the corresponding widths to the min-heap, after removing
from it the entries of all the eliminated edges of $\conv(Q')$.  The
root of the min-heap yields the best width for $S_2$ for the current
partition $\{Q,Q'\}$.
The algorithm requires presorting points by distance from $\ell$ and
then a linear number of balanced-search-tree and heap operations
(since the number of edges inserted is less than $2n$ and each cannot
be deleted more than once), for a total running time of $O(n \log n)$
for a fixed $e$, as claimed.

Working through the entire set $P$ (except for the endpoints of $e$), in
order of decreasing distance from $\ell$, growing $Q'$ and shrinking
$Q$, we obtain a sequence of fewer than~$n$ V-shapes which include
all the canonical minimum-width V-shapes covering $P$ with $e$ on its
outer boundary and two other points of $P$ lying on the opposite arm's
inner boundary (it may include non-canonical V-shapes as well, but it
is not difficult to check  
that every combination $(e,q,e',\ell_{e'})$ examined
by the algorithm yields a valid V-shape covering $P$, which is
sufficient for our purposes). 

To summarize, inner-outer type V-shapes can be handled in total
time $O(n^2\log n)$.

\paragraph{V-shapes of both-inner type.}
Now a covering V-shape~$V$ has points $a,b$ of $P$ on its inner ray
$X_1$ and points $c,d$ on its inner ray $X_2$; refer to Firgure~\ref{V4}; 
points $a,b,c,d$ are in convex position, in this counterclockwise order.  
\begin{figure}
\centering
\scalebox{0.5}{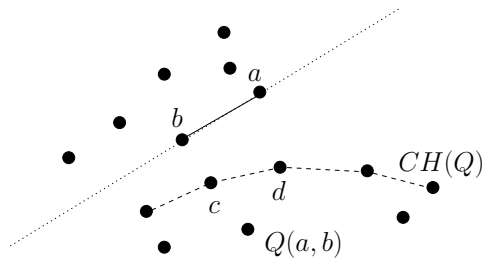}
\caption{Empty wedge defined by $ab$ and $cd$.}
\label{V4}
\end{figure}
It is known~\cite{EmptyPolygons} that there are at most $O(n^2)$ such
\emph{wedges} $W=W(a,b,c,d)$ determined by a quadruple of points
$a,b,c,d \in P$ and empty of points of $P$; note that $W$ completely
determines~$V$, and, given~$W$, one can construct the
corresponding~$V$ in $O(\log n)$ time,
so it is sufficient to enumerate all empty wedges~$W$.



For a pair $a,b \in P$, we compute all pairs $c,d$, so that
$W(a,b,c,d)$ is an empty wedge.  
Let $Q(a,b)$ be the set of all points of $P$ lying to the left of
the directed line $ab$.  

\begin{observation}
  $W(a,b,c,d)$, in the above notation, is an empty wedge if and only
  if line $cd$ supports $\conv(Q)$ and separates segment $ab$
  from $Q=Q(a,b)$ (and $a,b,c,d$ are in this counterclockwise order).
\end{observation}

Now enumerating all $k$~pairs $c,d$ for a fixed choice of $a,b$ can be
done in time $O((k+1) \log n)$, as follows.
While handling V-shapes of both-outer type
we constructed a data structure $D(P)$ which, for a given line (here
$ab$), produces a balanced search tree storing the convex hull of the
points of $P$ lying to one side of the line (here $Q=Q(a,b)$).  Using
$D(P)$, we find the point $z$ of $Q$ closest to the line $ab$ and
traverse the boundary of $\conv(Q)$ in both directions from $z$,
to list all $k$ edges $cd$ of $\conv(Q)$ satisfying the conditions
of the above observation.  
Since all such edges are consecutive, it is sufficient to examine
$k+2$ edges of $\conv(Q)$.
Repeating the procedure for all choices of $a,b$ and recalling that
the number of empty wedges is at most quadratic, we deduce that the
enumeration algorithm runs in time $O(n^2\log n)$.


\section{Maximum number of canonical minimum-width V-shapes}
\label{sec:lower-bound}

How far is our algorithm from optimality?  In Figure~\ref{fig:lower},
starting with the vertex set of two congruent regular $k$-gons, for a
suitably large~$k$, we sketch a construction of a set of $n$ points
with $\Theta(n^2)$ distinct covering minimum-width V-shapes.  The idea
is that a minimum-width covering V-shape would consist essentially of
two independently chosen minimum-width strips, each covering one
$k$-gon.  The point set is highly degenerate.  However, perturbing it
slightly yields a point set with $\Theta(n^2)$ canonical V-shapes with
width arbitrarily close to minimum possible.  This is an indication
that any algorithm that explicitly enumerates canonical covering
V-shapes may have to spend $\Omega(n^2)$ time on this input, thus it
is unlikely that any algorithm taking our approach can run much
faster.  On the other hand, for this specific input one can encode the
$\Theta(n^2)$ optimal V-shapes in $\Theta(n)$ space.  This leaves open
the possibility that identifying a single minimum-width covering
V-shape may still be possible in subquadratic worst-case time.

\begin{figure}
  \centering
  \includegraphics[scale=0.8]{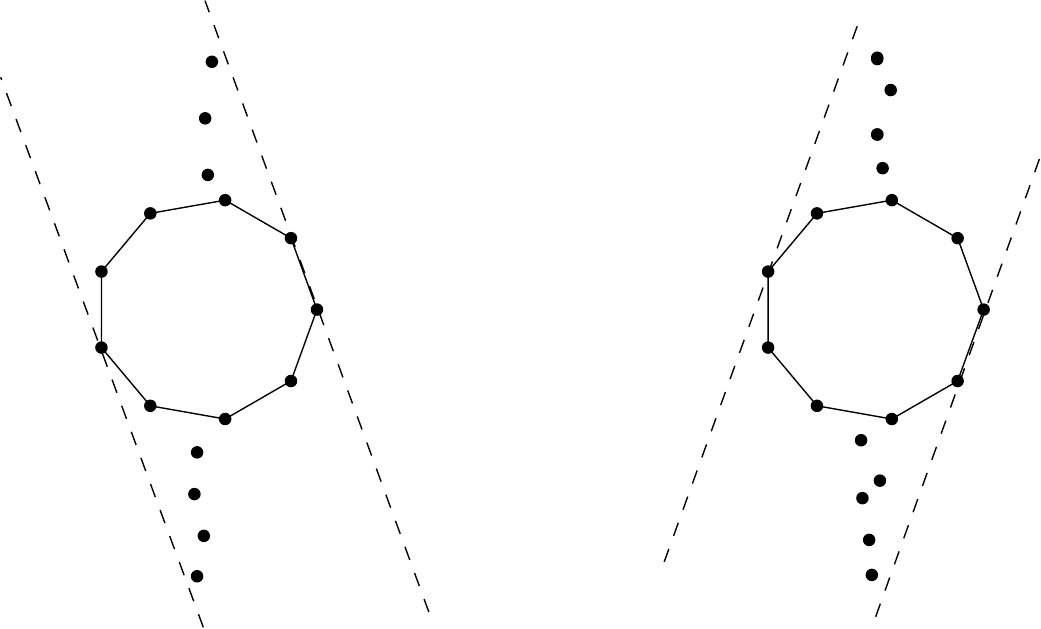}
  \caption{A sketch of a construction of a set with many minimum-width
    covering V-shapes.  Points below and above the two regular
    $k$-gons are added to raise the minimum width of any covering
    V-shape to that of the $k$-gons.}
  \label{fig:lower}
\end{figure}

\section{A 13-approximation algorithm}
\label{sec:13-approx}

Given a set of points $P$, let $w$ be the minimum value such that~$P$
can be covered by a V-shape of width~$w$.
We present an algorithm that computes a V-shape covering $P$ of width at most $13 w$ in time $O(n \log n)$.
For this purpose, we use the $O(n \log n)$ time 6-approximation algorithm for the 2-line-center problem  presented by Agarwal, Procopiuc, and Varadarajan \cite{2-line-center}.
Recall that the 2-line-center problem is the following: Given a set $P$ of $n$
points in $\bbbr^2$, cover $P$ by two congruent strips of minimum
width.  We start with the following observation which follows from the
fact that the union of the two strips of any V-shape covering $P$ contains $P$.

\begin{observation}
  \label{Obs:width}
  If $w'$ is the width of two congruent strips of minimum width covering $P$, $w' \leq w$.
\end{observation}

Our 13-approximation algorithm proceeds as follows.  Use the 6-approximation
algorithm of \cite{2-line-center} to compute two congruent
strips of width $w''$ that cover $P$, with $w' \leq w'' \leq 6w'$.
(It is possible that the two strips just computed are such
that a V-shape defined by them contains $P$.  In this case we return
that V-shape.  This clearly produces a 6-approximation, due to
Observation~\ref{Obs:width}.  In the remainder of this section, we
will assume that this is not the case, in other words, one of the two
strips has points of $P$ on both sides of it.)
Find the median lines $\ell_1$ and $\ell_2$ of
the strips.  For all points in each strip, project them orthogonally
onto $\ell_1$
and $\ell_2$ respectively (the points in the intersection of the
strips are duplicated and projected onto both $\ell_1$ and $\ell_2$).
Let $P'$ be the resulting set of projected points.  Compute an exact
minimum width V-shape $V'$ covering $P'$ (see Section~\ref{sec:two-lines})
in $O(n \log n )$ time.  The desired approximate V-shape $V$ is
obtained by widening $V'$ by $w''/2$ in all directions.

\begin{theorem}
This algorithm computes a 13-approximation of a minimum-width V-shape covering~$P$.
\end{theorem}
\begin{proof}
  Let $\Vbest$ be a minimum-width covering V-shape of $P$, $V'$ --- a
  minimum-width covering V-shape of $P'$, and $\Vapp$ --- the approximate
  covering V-shape computed by the algorithm.
  As the points of $P$ have been moved by a distance
  of at most $w''/2$ to form $P'$, $\width(V') \leq \width(\Vbest) + w''$.
  Since $\Vapp$ is a widened version of $V'$, it contains the points
  of $P$.  Moreover, $\width(\Vapp) \leq \width(V') + w'' \leq
  \width(\Vbest) + 2 w'' \leq w +12 w' \leq 13 w$ by
  Observation~\ref{Obs:width}. \hfill $\Box$
\end{proof}

\begin{uremark}
  Using the $(1+\eps)$-approximation
  algorithm of \cite{2-line-center} in place of their 6-approximation
  algorithm in our procedure, we can attain any approximation factor
  larger than 
  three for the minimum-width V-shape.  The running time remains $O(n
  \log n)$, with the constant of proportionality depending on the
  quality of the approximation.  We do not discuss this extension further,
  since we present our own $(1+\eps)$-approximation algorithm for the
  problem in Section~\ref{sec:one+eps}.
\end{uremark}

\subsection{Minimum-width V-shape for points on two lines}
\label{sec:two-lines}

We now describe how to compute the minimum-width V-shape~$V'$ covering
a given point set~$P'$ contained in the union of two lines $\ell_1$, $\ell_2$
in the plane.  Put $z:=\ell_1 \cap \ell_2$.

Let $P_1':=P' \cap \ell_1$ and $P_2':=P'\cap \ell_2$.  If $P'$ is
covered by a zero-width $V$-shape, which is easy to check, we are
done.  From now on we assume that this is not the case, i.e.,
that $\ell_1$, $\ell_2$ do not already form a V-shape
containing $P'$, so $\ell_1$ separates some two points of $P_2'$ and/or $\ell_2$
separates some two points of $P_1'$.  The convex hull $\conv(P')$ has
three or four vertices.  Moreover, by reasoning similar to that of
Section~\ref{sec:algorithm}, the outer boundary of $V'$ 
contains two, three, or four vertices of $\conv(P')$ (in the case where 
an outer ray is contained in $\ell_1$ or $\ell_2$, we consider only
the extreme points).  Before describing
how we handle these cases, we need a technical lemma.

\begin{lemma}
  \label{lem:cut}
  Given a line partitioning $P'$ into $P_r'$, $P_\ell'$ and given
  their convex hulls $\conv(P_r')$, $\conv(P_\ell')$, the minimum-width
  canonical V-shape~$V'$ of $P'$ containing $P_r'$ in one strip and $P_\ell'$ in
  the other can be computed in constant time; some points
  of $P'$ might lie in both strips of $V'$.
\end{lemma}

\begin{proof}
  The convex hulls $\conv(P_r')$ and $\conv(P_\ell')$ have at most four
  vertices each. It must be the case that the
  boundary of one arm of $V'$ contains an edge $e_r$ of $\conv(P_r')$ or
  an outer common tangent of $\conv(P_r')$ and $\conv(P_\ell')$, and the
  other arm boundary contains an edge $e_\ell$ of $\conv(P_\ell')$ or an
  outer common tangent to $\conv(P_r')$ and $\conv(P_\ell')$.  There is
  a constant number of possible pairs of such edges.  Let $S(X,e)$ be
  the minimum-width strip covering a set $X$ and parallel to $e$.  For
  each such pair of edges $e_r,e_\ell$, check whether $S(P_r',e_r)$
  and $S(P_\ell',e_\ell)$ form a V-shape.  Return the canonical V-shape of
  minimum width among all V-shapes so generated. \hfill $\Box$
\end{proof}

Now we consider the different types of canonical V-shapes
covering~$P'$ and describe how to find a minimum-width V-shape of each
type.

\emph{Case 1: An outer bounding ray of $V'$ contains an edge $e$ of $\conv(P')$.}
Let $\ell$ be the line containing $e$.  For all points $p$ of $P'$,
draw a line $\ell_p$ through $p$ and parallel to~$\ell$.  Apply
Lemma~\ref{lem:cut} to (the partition induced by) $\ell_p$.  This can be
implemented to run in overall time $O(n\log n)$.

\smallskip

In the remaining cases, each of the outer rays of $V'$ contains
precisely one vertex of $\conv(P')$ and each inner ray contains two
points of $P'$.

\emph{Case 2: An inner ray of $V'$ lies on $\ell_1$ or $\ell_2$.}
Suppose an inner ray of $V'$ is contained in $\ell_1$.
Draw two lines parallel to $\ell_1$ and very close to it,
one to the left of $\ell_1$, one to the right of $\ell_1$.
Apply Lemma~\ref{lem:cut} to each of these two lines.

\emph{Case 3: Point $z=\ell_1\cap\ell_2$ lies between the two arms of $V'$.}
Draw the two lines passing through $z$ and bisecting the angles
between $\ell_1$ and $\ell_2$.  
Apply Lemma~\ref{lem:cut} to each of these two
lines.

\emph{Case 4: Point $z$ is inside one arm of $V'$.}
For each pair of consecutive points $p,q \in P'$ on $\ell_1$ or
on $\ell_2$ not separated by $z$, apply Lemma~\ref{lem:cut} to the
perpendicular bisector of the segment $pq$.

Now we argue that the last procedure returns the best minimum-width
V-shape $V'$ of $P'$ with two points on its outer boundary and $z$ in
one of its arms, correctly handling case~4 and thereby concluding our description.

Let $x$ and $y$ be the vertices of $V'$.  For ease of presentation,
rotate the entire picture so that $y$ lies below $x$; refer to figure~\ref{V8}.
Let $s_3$ be the point of $P$ on $Y_1$, $s_1$, $s_2$ be the points on $X_1$, with $s_1$ closer to $x$ than $s_2$.
Similarly let $s_4$ be the point on $Y_2$, $s_5$, $s_6$ be the points on $X_1$, with $s_5$ closer to $x$ than $s_6$. As $\ell_1$ and $\ell_2$ don't intersect between the two arms of $V'$, $s_1$ and $s_5$ lie on one line, $s_2$ and $s_6$ lie on the other line.
Let $s_1$ and $s_5$ lie on $\ell_1$, and $s_2$ and $s_6$ lie on
$\ell_2$, without loss of generality.

\begin{figure}
\centering
\subfloat[Three boundary points on $\ell_1$,$\ell_2$.]{\label{V8} \scalebox{0.4}{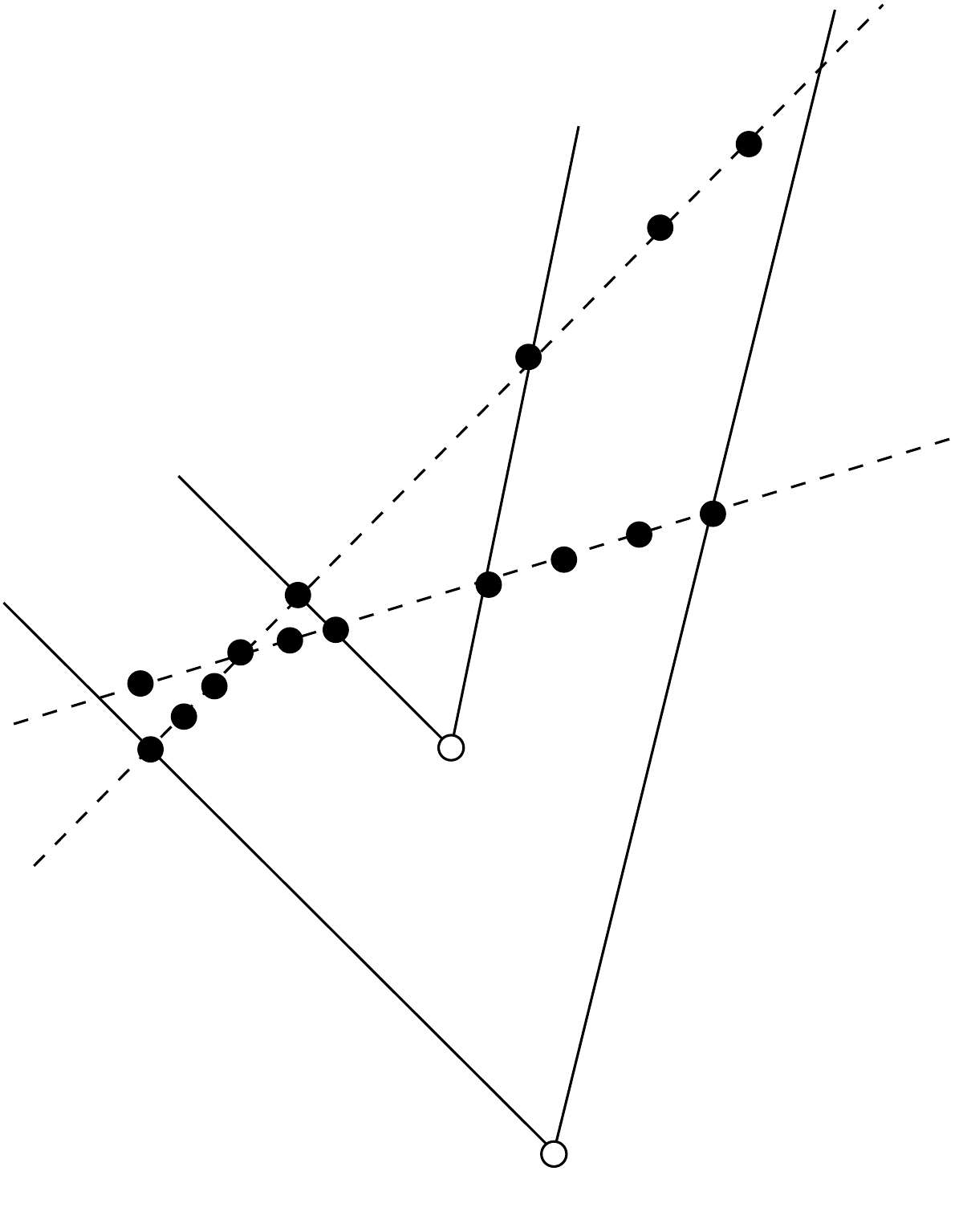}}
\subfloat[ Four boundary points on $\ell_2$]{\label{V9} \scalebox{0.4}{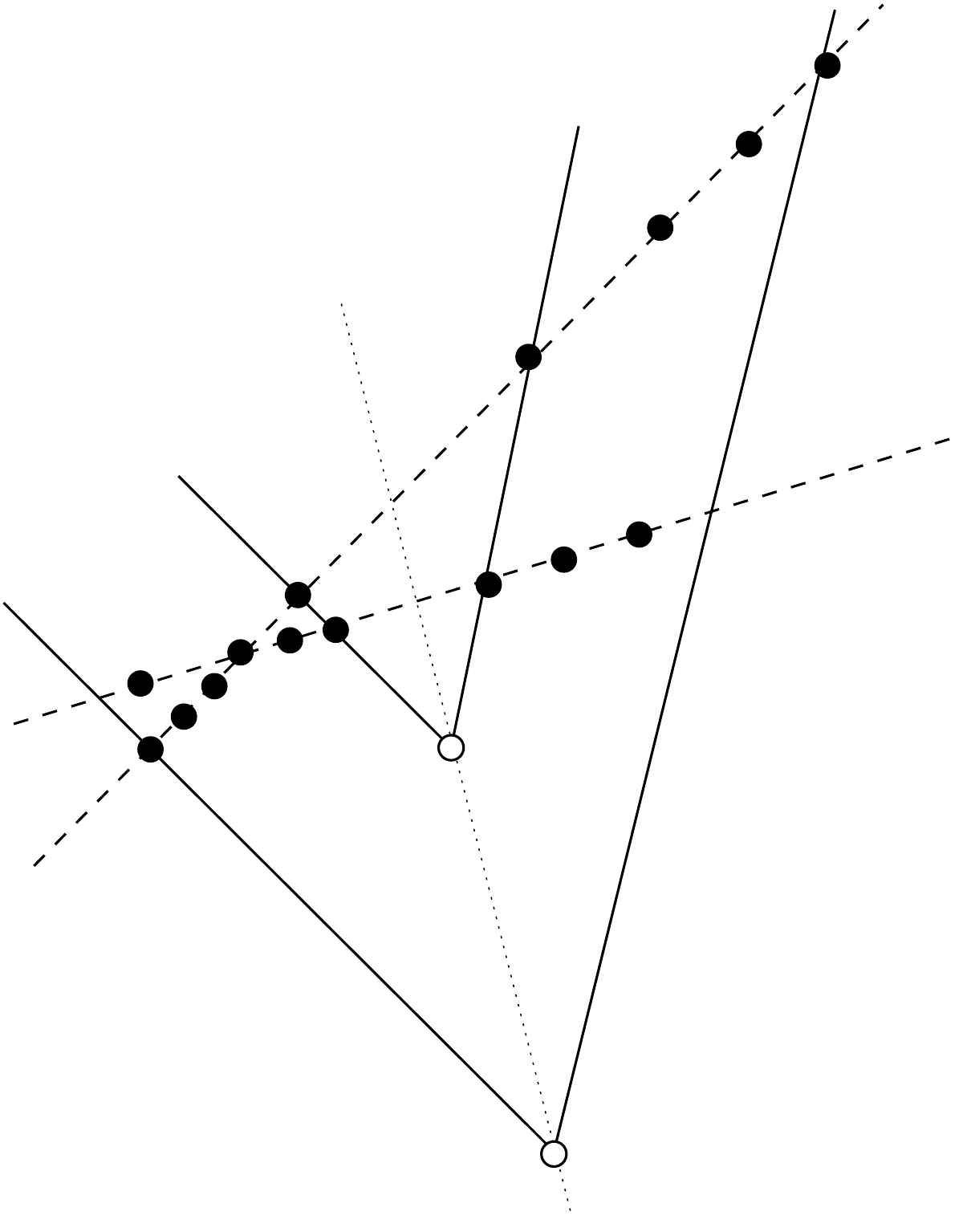}}
\newline
\subfloat[Four boundary points on $\ell_1$.]{\label{V10} \scalebox{0.4}{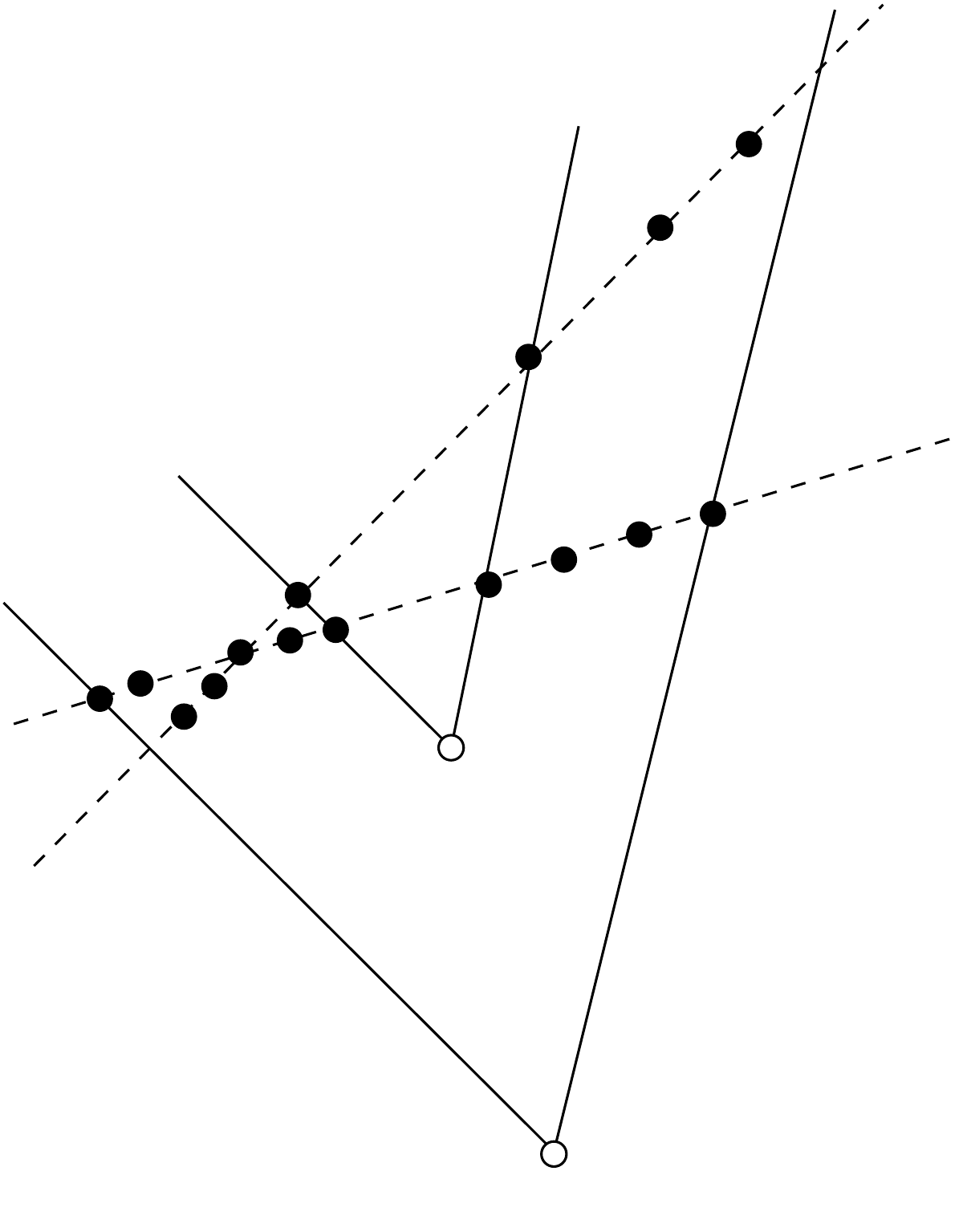}}
\caption{$V'$}
\end{figure}

The three points of $P'$ on one arm boundary cannot lie on the same
line $\ell_i$, as they form a triangle.  Therefore either each line $\ell_i$ contains
three boundary points belonging to three different boundary rays, or one
line contains four boundary points from four boundary rays, and the
other line contains two boundary points from the two inner rays.

We consider all possible cases: 
\begin{enumerate} \itemsep 0pt \parsep 0pt \parskip 0pt
\item $\ell_1$ contains $s_1$, $s_5$, and $s_4$, and $\ell_2$ contains $s_3$, $s_2$, and $s_6$ ($\ell_1$ contains $s_1$, $s_3$, and $s_5$, and $\ell_2$ contains $s_2$, $s_4$, and $s_6$ is a symmetric case).
\item $\ell_2$ contains four boundary points, $\ell_1$ contains two boundary points.
\item $\ell_2$ contains two boundary points, $\ell_1$ contains four boundary points.
\end{enumerate}
In each case, we prove either that the configuration
of the points is impossible, or that the angles
$\alpha_1$ and $\alpha_2$ (see
figure~\ref{V8}) between $\ell_1$ and the inner boundary rays are
acute.  When the angles are acute, the perpendicular bisector of $s_1 s_5$
separates the points of $P'$ belonging to each arm of $V'$.  Therefore
Lemma~\ref{lem:cut} can be applied and our handling of case~4 is justified.

We consider the cases in turn:
\begin{enumerate} \itemsep 0pt \parsep 0pt \parskip 0pt
\item $\ell_1$ contains $s_1$, $s_5$, and $s_4$, and $\ell_2$ contains
  $s_3$, $s_2$, and $s_6$
(see figure~\ref{V8}).
As the angle $\angle s_4 s_5 s_6$ is acute, its opposite angle $\alpha_2$ is acute. 
What is left to prove is that $\alpha_1$ is acute as well.
The angle $\angle s_1 s_2 s_3$ is acute, hence $\alpha_5$ is obtuse,
and so is $\alpha_6$.
But the angle $\angle s_5 s_6 s_4$ is acute, hence $\ell_2$ does not intersect $\ell_1$ in the right arm of $V'$, so $\ell_1$ and $\ell_2$ intersect in the left arm.
More precisely $\ell_1$ intersects the segment $s_3 s_2$.
The opposite angle of $\alpha_1$ is smaller than the angle $\angle s_2 s_1 s_3$, which is acute.
Therefore 
so is $\alpha_1$.
\item $\ell_2$ contains four boundary points, $\ell_1$ contains two boundary points.
Let $\ell_{s_2 s_3}$ be the line containing $s_2, s_3$, $\ell_{s_6 s_4}$ be the line containing $s_6,s_4$, and $\ell_{xy}$ be the line containing $x,y$ (see figure~\ref{V9}).
By Corollary~\ref{Cor:angle}, as the angle $\angle s_3 s_2 s_1$ is acute, $\ell_{s_2 s_3}$ forms an acute angle $\alpha_3$ with $\ell_{xy}$.
Similarly, as the angle $\angle s_5 s_6 s_4$ is acute, $\ell_{s_6 s_4}$ forms an acute angle $\alpha_4$ with $\ell_{xy}$.
But as $\ell_{s_3 s_2} = \ell_{s_6 s_4}$, $\alpha_3$ and $\alpha_4$ are supplementary, a contradiction.
\item $\ell_2$ contains two boundary points, $\ell_1$ contains four boundary points (see figure~\ref{V10}).
By Corollary~\ref{Cor:angle}, the angles $\angle s_2 s_1 s_3$ and $\angle s_4 s_5 s_6$ are acute, therefore their opposite angles $\alpha_1$ and $\alpha_2$ are acute as well.
\end{enumerate}


\section{A $(1+\eps)$-approximation algorithm}
\label{sec:one+eps}

In this section we describe how to construct, given a point set $P$
and a real number $\eps>0$, a V-shape $V$ covering~$P$, with
$\width(V)\leq(1+\eps)\wopt$, where $\wopt$ is the width of a
minimum-width V-shape covering $P$.

We start by recalling the notion of an anchor pair used in
\cite{2-line-center}.  Given a V-shape~$V$ covering $P$, fix one of
the strips of $V$, say $S_1$.  We say that a pair of points $p,q \in P
\cap S_1$ is an \emph{anchor pair}, if $\dist(p,q) \geq \mathrm{diam}(P \cap
S_1)/2$.  Lemma~3.3 in~\cite{2-line-center} describes how to identify
at most~11 pairs of points in~$P$, such that, for any two-strip cover
of~$P$, at least one of the pairs is an anchor pair for one of the
strips; the algorithm requires $O(n \log n)$ time.  As covering by a
V-shape is a special case of covering by two strips, the definition
and the algorithm apply here as well.

We show how to, given a potential anchor pair~$p,q$, construct a
$(1+\eps)$-approximation of the minimum-width V-shape covering $P$ for
which $p,q$ is an anchor pair.  More precisely, below we prove
\begin{lemma}
  \label{lem:fixed-anchor-pair}
  Given a potential anchor pair $p,q \in P$, we can construct, in time
  $O((n/\eps)\log n + (n/\eps^{3/2})\log^2(1/\eps))$, a V-shape
  covering $P$, of width at most $1+\eps$ times the minimum width of
  any V-shape covering $P$ for which $p,q$ is an anchor pair.
\end{lemma}

Applying this procedure at most~11 times, we obtain our desired
approximation algorithm:
\begin{theorem}
  A V-shape covering $P$ and of width at most $(1+\eps)\wopt$ can be
  constructed in time $O((n/\eps)\log n + (n/\eps^{3/2})\log^2(1/\eps))$.
\end{theorem}

We first prove that it is sufficient to consider those V-shapes~$V$
with anchor pair $p,q$, for which the strip containing $p,q$ has one
of a small set of fixed directions.  Setting $\beta:=
\sin^{-1}\min\{\eps\cdot\width(V)/(6 d(p,q)),1\}$ and
$\gamma:= \beta + \sin^{-1}(\min\{1,\width(V)/d(p,q)\})$, we
prove the following

\begin{lemma}
  \label{lem:rotate-strip}
  Let V-shape $V$ cover $P$, and let $p,q$ be an anchor pair for
  $S_1(V)$.  Rotating $S_1(v)$ by an angle at most $\beta$ does not
  increase the width of the V-shape by more than a factor of
  $1+\eps/3$, and the angle between $pq$ and the direction of the
  rotated strip cannot exceed $\gamma$.
\end{lemma}

\begin{proof}
  Put $w:=\width(V)$.  Let $B$ be the minimum bounding box of $P \cap
  S_1$.  More precisely, it is the shortest rectangle cut out of $S_1$
  by two lines perpendicular to $S_1$ and containing $P \cap S_1$;
  refer to Figure~\ref{fig:rotated}.  Let $s$ and $t\leq w$ be the
  length (along the axis of $S_1$) and width of $B$, respectively.
  \begin{figure}
 \centering
    \includegraphics[width=0.5\textwidth]{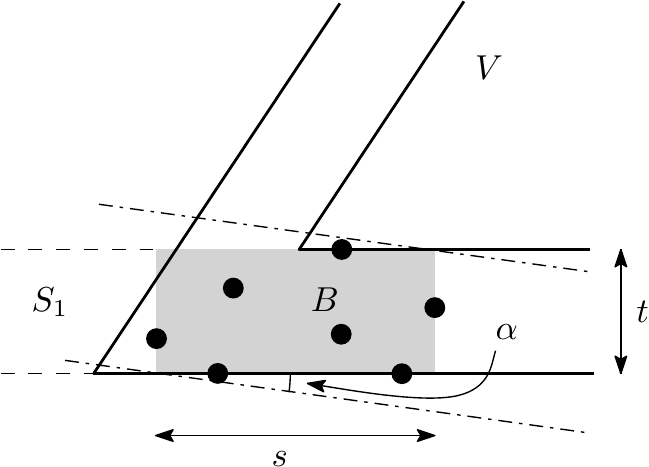}
    \caption{Rotation by $\alpha$ does not change the width by much.}
    \label{fig:rotated}
 \end{figure}
  Let $S'_1$ be the minimal parallel strip containing $B\cap V$, whose
  direction is $\alpha \leq \beta$ away from that of $S_1$ (there are two choices for $S'_1$, corresponding to
  rotating clockwise and counterclockwise; only one is shown; the
  argument applies to both cases).  
  Then
  \begin{align*}
    \width(S'_1) & \leq s \sin \alpha + t \cos \alpha 
    \leq  2d(p,q) \sin \alpha + w \\
    & \leq w ( 1 + 2 \frac{d(p,q)}{w} \sin\alpha ) \leq t (1 + \eps/3),
  \end{align*}
  since $\sin \alpha \leq \sin\beta \leq \eps w/(6 d(p,q))$.  Now
  replace $S_1$ by $S_1'$ to obtain a new V-shape $V'$ covering~$P$.
  Its width is $\min\{\width(S_1'),\width(S_2)\}\leq(1+\eps/3)w$, as
  claimed.

  Observe that in the above construction, the angle between $pq$ and
  the direction of $S_1'$ cannot exceed
  \[\alpha+\sin^{-1}(\min\{1,t/d(p,q)\}) \leq \beta +
  \sin^{-1}(\min\{1,\width(V)/d(p,q)\}) = \gamma.\] \hfill $\Box$
  \end{proof} 

We conclude that enumerating all V-shapes that contain $p,q$ in their
strip $S_1$ and whose directions are (a)~at most $\gamma$ away from
that of $d(p,q)$ and (b) spaced at most $\beta$ apart, would yield a
V-shape whose existence is claimed in
Lemma~\ref{lem:fixed-anchor-pair}.  The number of directions to be
tested is at most $O(\gamma/\beta)=O(1/\eps)$.

Given a candidate anchor pair $p,q$, the algorithm proceeds by
starting with the direction $pq$.  Since we need not consider V-shapes
whose width is larger than the approximate width $\wapp$ computed in
Section~\ref{sec:13-approx} (this is where the 13-approximation
algorithm is used to bootstrap our $(1+\eps)$-approximation), we
replace $\width(V)$ by the smaller $\wapp/13$ in the
definition of $\beta$ above and by the larger $\wapp$ in the
definition of $\gamma$, thereby erring on the conservative side in
each case.  Having computed (conservative estimates of) $\beta$ and
$\gamma$, we enumerate the $O(1/\eps)$ directions of the form
$\theta_i:=\theta_{pq}+i\beta$, where $\theta_{pq}$ is the direction
of $pq$ and $i$ is an integer ranging from $-\lceil
\gamma/\beta\rceil$ to $\lceil \gamma/\beta\rceil$.  It remains to
explain how to deal with one such direction~$\theta:=\theta_i$.

\begin{lemma}
  \label{lem:apx-one-direction}
  One can compute a canonical V-shape $V$ covering $P$ with one arm in
  given direction $\theta$ and width at most $1+\eps/3$ times the
  minimum width of any such V-shape, in time $O(n\log n +
  (n/\eps^{1/2})\log^2(1/\eps))$.
\end{lemma}

\begin{proof}
  We use an approach similar to that of the inner-outer case of our
  exact algorithm 
  with a slight twist.  

  Let $\ell$ be a line in direction~$\theta$ supporting~$\conv(P)$.  We
  again let $q$ be the furthest point from $\ell$ in $Q:=P\cap S_1$
  and let $Q':=P \setminus Q$.  When $q$ is fixed, the minimum-width
  V-shape is determined by the minimum-width strip $S_2$ covering $Q'$
  and not ``splitting'' $P$, i.e., such that it does not have points of
  $P$ on both sides of it.  It is easy to ensure that $S_2$ does not
  split $P$ by observing that a direction of $S_2$ lying between the
  directions of the common outer tangents to $\conv(Q)$ and $\conv(Q')$ is
  never useful.  Depending on the side where the lines supporting
  these tangents cross, a minimal strip $S_2$ covering $Q'$ and lying
  in the range between them either crosses $Q$ (and therefore~$P$) or
  completely covers $Q$ (and therefore~$P$).  In the former case,
  $S_1$ and $S_2$ do not form a legal V-shape covering $P$ and in the
  latter they form a covering V-shape with one empty strip, which 
  never yields minimum width by reasoning as in Lemma~\ref{lem:no-empty-strip}. 

  The width of the resulting V-shape is the maximum of $\dist(q,\ell)$
  and (the restricted) $\width(S_2)$.  The algorithm proceeds by processing points $q$
  in order of decreasing distance to $\ell$, keeping track of
  $\dist(q,\ell)$ and a \emph{coreset} for $Q'$, which is a subset of
  $Q'$ with the property that its directional width, in every
  direction, is at least $1-\eps/3$ that of $Q'$ (and, expanding the
  corresponding minimal strip containing the subset by
  a factor of $1+\eps/3$, we get a strip covering $Q'$).  Chan~\cite{chan04}, in
  Theorem~3.7 and remarks in Section 3.4, describes a streaming
  algorithm that maintains an $O(1/\eps^{1/2})$-size coreset at an
  amortized cost of $O((1/\eps^{1/2})\log^2(1/\eps))$ per insertion.
  For a fixed $q$, we go through the coreset (after computing its
  convex hull, if necessary), and determine the narrowest strip
  covering it and satisfying our angle constraints.  The maximum of
  that and $\dist(q,\ell)$ gives the width of the minimum-width
  V-shape whose boundary passes through~$q$.\footnote{%
    More precisely, $q$ lies on the boundary of $S_1$ and may not even appear
    on the boundary of $V$. However, as before, all V-shapes we
    examine are valid and cover~$P$, and the desired approximating
    V-shape is among them, which is sufficient.}
  The amortized cost per
  point is dominated by the $O((1/\eps^{1/2})\log^2(1/\eps))$ cost of
  insertion.  Together with presorting points by distance from $\ell$,
  the total cost is then $O(n\log n + (n/\eps^{1/2})\log^2(1/\eps))$. \hfill $\Box$
\end{proof}

Combining Lemmas
\ref{lem:rotate-strip}~and~\ref{lem:apx-one-direction} yields the
procedure claimed in Lemma~\ref{lem:fixed-anchor-pair} and thereby
completes our description of the $(1+\eps)$-approximation algorithm.


\section{Concluding remarks}\label{sec:conclusion}

As mentioned in the introduction, this work was inspired by research
on curve fitting, in the situations where a curve takes a sharp turn.
Besides the exact and approximate versions of the problem studied
above, it would be natural to investigate a variant that can handle 
a small number of outliers.  A natural ``peeling''
approach to the problem would be to eliminate the points defining the
optimal V-shape found by our exact algorithm and trying again.
However, it is easy to construct an example of a point set in which
removing a single point \emph{not} appearing on the boundary of the
minimum-width covering V-shape significantly reduces the width of the
optimum V-shape.

Are there natural assumptions (perhaps in the spirit of ``realistic
input models''~\cite{realistic-input-models} or in the form of
requiring reasonable sampling density) that would be relevant for the
curve-fitting problem, and that would make finding the minimum-width covering
V-shape easier?

Returning to the problem studied in the paper, is it possible to find
an exact minimum-width covering V-shape in subquadratic time?  Is the
problem \textsc{3sum}-hard?  

Is it possible to speed up the approximation algorithm, improving the
dependence of its running time on $\eps$?
Is time $O( n + f(\frac{1}{\eps}))$ achievable?

Finally, we would like to point out that there are other
``reasonable'' definitions for a V-shape, if the goal is to
approximate a sharp turn of a curve: One can imagine defining a
V-shape as the Minkowski sum of a disk with the union of two rays
emanating from a common point as in~\cite{GKS}.  
The width of the
V-shape would be the diameter of the disk.  Can the exact algorithm
from~\cite{GKS} be sped up?  Is there a faster approximation
algorithm?  Is this version of the problem better suited for curve
fitting?

\section*{Acknowledgments}

We thank Piyush Kumar for suggesting the problem studied in this
paper, Pankaj K. Agarwal for extremely useful pointers, Alfredo Hubard
for several helpful conversations, and Sariel Har-Peled for sharing
his wisdom and dispelling some confusion about approximation.



\end{document}